\documentclass{article}
\usepackage{amsmath,amssymb,amsthm}
\usepackage[dvipdfmx]{graphicx} 
\usepackage{float}
\usepackage{textcomp}
\theoremstyle{definition}
\newtheorem{definition}{Definition}
\newtheorem{example}{Example}
\theoremstyle{plain}
\newtheorem{proposition}{Proposition}
\newtheorem{lemma}[proposition]{Lemma}
\newtheorem{theorem}[proposition]{Theorem}

\begin{document}
\title{A Solution to Yamakami's Problem on Advised Context-free Languages
}
\author{Toshio Suzuki
\thanks{
This work was partially supported by Japan Society for the Promotion of Science (JSPS) KAKENHI (C) 22540146.}
\\ 
Department of Mathematics and Information Sciences, \\ 
Tokyo Metropolitan University, \\ 
Minami-Ohsawa, Hachioji, Tokyo 192-0397, Japan\\
%
toshio-suzuki@tmu.ac.jp
}
\date{\today}

\maketitle              

\begin{abstract}
Yamakami (\textit{Theoret. Comput. Sci.}, 2011) studies context-free languages with advice functions. 
Here, the length of an advice is assumed to be the same as that of an input. 
Let CFL and CFL/$n$ denote the class of all context-free languages and that with advice functions, respectively. 
We let CFL(2) denote the class of intersections of two context-free languages. 
An interesting direction of a research is asking how complex CFL(2) is, relative to CFL. 
Yamakami raised a problem whether there is a CFL-immune set in CFL(2) - CFL/$n$. The best known so far is that LSPACE - CFL/$n$ has a CFL-immune set, 
where LSPACE denotes the class of languages recognized in logarithmic-space. 
We present an affirmative solution to his problem. 
Two key concepts of our proof are the nested palindrome and Yamakami's swapping lemma. 
The swapping lemma is applicable to the setting where the pumping lemma (Bar-Hillel's lemma) does not work. Our proof is an example showing how useful the swapping lemma is. 

\vspace{\baselineskip}

Keywords: context-free language; push-down automaton; advice function; non-uniform complexity class; immune set. 
\end{abstract}

\section{Introduction}

The regular languages have beautiful closure properties. For example, given two regular languages, their intersection is a regular language. Nevertheless, in the studies of programming languages, most of important languages are not regular. The same holds in the studies of  formal models of natural languages. In the case of classes larger than the regular languages,  closure properties are more difficult than the regular cases. 

In particular,  given two context-free languages, their intersection is not necessarily context-free. For a positive integer $k$, we consider the intersection of $k$ context-free languages, and let CFL($k$) denote the class of all such intersections. CFL(1) is CFL, the class of all context-free languages. It is known that CFL($k$) is a proper subset of CFL($k+1$). 

How complex is CFL($k+1$), relative to CFL($k$)? 
An interesting observation is given by Flajolet and Steyaert \cite{FS1974}. 
Let $L_{3\mathrm{eq}}$ denote the set of all strings of the form $0^n 1^n 2^n$ where $n$ is a natural number. 
It is easily seen that $L_{3\mathrm{eq}}$ belongs to CFL($2$). 
Flajolet and Steyaert observed that $L_{3\mathrm{eq}}$ is CFL-immune. 

Here, an immune set is a key concept in the classical recursion theory, namely in Post's problem. 
Later, immune sets relative to complexity classes are studied in the complexity theory \cite{YS2005}. Given a class ${\cal C}$ of languages, an infinite  language $A$ is \emph{${\cal C}$-immune}\/ if no infinite subset of $A$ belongs to ${\cal C}$. 

\vspace{\baselineskip}

\noindent
\textbf{Yamakami's problem}~
Yamakami \cite{Y2011} raised a problem whether there is a CFL-immune set in CFL(2) - CFL/$n$. 

\vspace{\baselineskip}

Here, the $n$ of CFL/$n$ denotes the length of an input. 

\begin{definition}  \label{def:cslashntyl} (Tadaki et al. \cite{TYL2005})
Given a class ${\cal C}$ of languages, we define ${\cal C}/n$ as follows. 
Suppose that $L$ is a language over an alphabet $\Sigma$. Suppose that $\Gamma$ is another alphabet. 
We introduce an extended alphabet $\left[ \begin{array}{c} \Sigma \\ \Gamma \end{array} \right]$. It consists of all symbols of the form $\left[ \begin{array}{c} x \\ a \end{array} \right]$ for $x \in \Sigma$ and $a \in \Gamma$. 
Given two strings $x=x_1 \cdots x_n \in \Sigma^n$ and $a = a_1 \cdots a_n \in \Gamma^n$ of the same length, we let $\left[ \begin{array}{c} x \\ a \end{array} \right]$ denote the string 
$\left[ \begin{array}{c} x_1 \\ a_1 \end{array} \right] \cdots \left[ \begin{array}{c} x_n \\ a_n \end{array} \right] \in \left[ \begin{array}{c} \Sigma \\ \Gamma \end{array} \right]^n$.

A language $L$ belongs to ${\cal C}/n$ if and only if there exist a language $L^{\prime} \in {\cal C}$ and a function $h : \mathbb{N} \to \Gamma^\ast$ such that for every $x \in \Sigma^\ast$,  the length of $h(|x|)$ is the same as that of $x$ and the following holds. 
\[
x \in L \,\, \Leftrightarrow \,\,
\left[ \begin{array}{c} x \\ h(|x|) \end{array} \right] \in L^{\prime}
\]
   
Then, $h$ is called an \emph{advice function}. $h(n)$ is the advice at length $n$. 
\qed
\end{definition}

Motives for the problem of Yamakami are the following examples on DCFL (the class of languages accepted by deterministic pushdown automata) and REG (the regular languages). 

(i) Let $L_{\mathrm{eq}}$ denote the set of all strings of the form $0^n 1^n$ ($n \geq 1$). 
Then $L_{\mathrm{eq}}$ belongs to DCFL $\cap$ REG$/n$ and is REG-immune \cite{FS1974}. 

(ii) Let Pal${}_\sharp$ denote the palindromes whose center symbol is a special symbol $\sharp$. Then Pal${}_\sharp$ belongs to DCFL $-$ REG$/n$ and is REG-immune \cite{Y2011}. 

The CFL-immune set $L_{3\mathrm{eq}}$ given in \cite{FS1974} belongs to 
CFL(2) $\cap$ CFL$/n$. Thus, it is natural and interesting to ask whether there is a CFL-immune set in CFL(2) $-$ CFL$/n$. 

The best known so far is that LSPACE - CFL/$n$ has a CFL-immune set \cite{Y2011}, 
where LSPACE denotes the class of languages recognized by deterministic Turing machines with a single read-only input tape and a logarithmic-space bounded work tape. 

What is the difficult point in the problem of Yamakami? 
A classical method of showing that a language is not context-free is the pumping lemma for CFL (Bar-Hillel's lemma \cite{BPS1961}). However, the pumping lemma destroys the advice $h(n)$. 

Our main theorem is an affirmative solution to the problem of Yamakami. 
Two key concepts of our proof are the nested palindrome and Yamakami's swapping lemma. Our test language is introduced in section~\ref{section:testlanguage}. 
The swapping lemma is applicable to the setting where the pumping lemma does not work. Our proof is an example showing how useful the swapping lemma is. 
We show our main theorem in section~\ref{section:main}. 

There is a study of advised context-free languages earlier than Tadaki et al. \cite{TYL2005}. 
In the setting of  Damm and Holzer \cite{DH1995}, the arrangement of the advice and the input is serial. 
In section~\ref{section:serial}, we show that the same result as our main theorem holds for the advised language class of Damm and Holzer. 

\section{Preliminaries}

\subsection{Notation}

For two sets $A$ and $B$, their \emph{difference} $A - B$ is $\{ x \in A : x \not\in B \}$. 
$A \subset B$ denotes that $A$ is a subset of $B$; $A$ may equal to $B$. 
$\mathbb{N}=\{ 0,1,2, \ldots \}$ is the set of all natural numbers. 
For a real number $x$, $\lceil x \rceil$ denotes the minimal natural number $n \geq x$. 

The empty string is denoted by $\lambda$. 
An \emph{alphabet} denotes a finite set of characters. 
For an alphabet $\Sigma$, the set of all strings is denoted by $\Sigma^\ast$. 
We let $\Sigma^+$ denote $\Sigma^\ast - \{ \lambda \}$. 
Given a string $w$, its \emph{length} $|w|$ denotes the total number of occurrences of characters. 
The \emph{reverse} of $w=w_1 \cdots w_n$, where $n=|w|$, is $w_n \cdots w_1$.  
The reverse of $w$ is denoted by $w^R$. 

REG (CFL, respectively) is the class of all regular (context-free) languages.
Suppose that ${\cal C}$ is a given class of languages such as REG or CFL. 
An advised class ${\cal C}/n$ is defined as in Introduction. 

\begin{figure}[h]
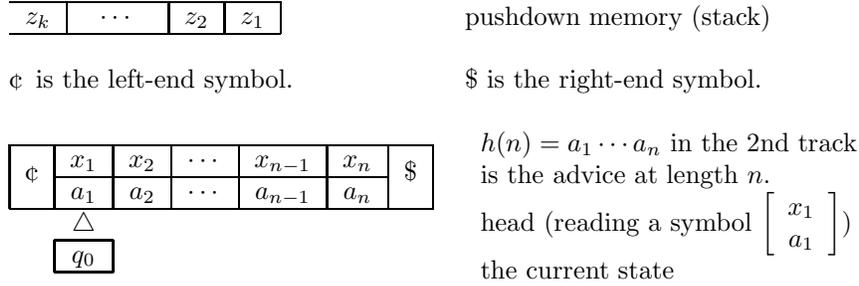

\begin{center}
\begin{tabular}{ll}
	\begin{tabular}{c|c|c|c|}
	\cline{1-4}
	$z_k$ &  \,\, $\cdots$ \,\, & $z_2$ & $z_1$\\ 
	\cline{1-4}
	\end{tabular}
	 & pushdown memory (stack) 
\\
& 
\\
	\textcent \, is the left-end symbol. & \$ is the right-end symbol.  
\\
&
\\
	\begin{tabular}{|c|c|c|c|c|c|c|}
	\cline{1-7}
	\raisebox{-5pt}[0cm][0cm]{\textcent} & $x_1$ &  $x_2$ & $\cdots$ & $x_{n-1}$ & $x_n$ & \raisebox{-5pt}[0cm][0cm]{\$}\\ 
	\cline{2-6}	
	& $a_1$ & $a_2$ & $\cdots$ & $a_{n-1}$ & $a_n$ & \\ 
	\cline{1-7}
	\multicolumn{1}{c}{} & \multicolumn{1}{c}{$\bigtriangleup$} & \multicolumn{5}{c}{} \\
	\cline{2-2}
	\multicolumn{1}{c}{} & \multicolumn{1}{|c|}{$q_0$} & \multicolumn{5}{c}{} \\
	\cline{2-2}
	\end{tabular}
&
	\begin{tabular}{l}
	$h(n)=a_1 \cdots a_n$ in the 2nd track \\
	is the advice at length $n$. \\
	head (reading a symbol $\left[ \begin{array}{c} x_1 \\ a_1 \end{array} \right]$) \\
	the current state
	\end{tabular}
\end{tabular}
\caption{A non-deterministic pushdown automaton with an advice \label{fig:npda}}
\end{center}
\end{figure}

The class CFL/$n$ is characterized by \emph{non-deterministic pushdown automata with an advice function} (Fig.~\ref{fig:npda}). 
It has a one-way read-only input tape and a pushdown memory (stack). 
The input tape has two tracks. An input is given on the first track. 
The advise at the length of the input is given on the second track. 
Then the automaton works as a non-deterministic automaton over the alphabet $\left[ \begin{array}{c} \Sigma \\ \Gamma \end{array} \right]$. 

${\cal C}(2)$ is the class of all languages that are intersections of two elements of ${\cal C}$. 
We have REG(2) = REG. On the other hand, CFL is a proper subset of CFL(2). 

An infinite language $L$ is \emph{${\cal C}$-immune} if $L$ does not have an infinite subset that belongs to ${\cal C}$. 

\subsection{The Swapping Lemma for Context-free Languages}

Suppose that $n$ is a positive integer, $S$ is a set of strings of length $n$, 
$i$ and $j$ are natural numbers such that $i+j \leq n$, 
and that $u$ is a string over $\Sigma$ of length $j$. 
Yamakami \cite{Y2009} defines a subset $S_{i,u}$ of $S$ as follows. 
\[
S_{i,u} = \{ v_1 \cdots v_n \in S : v_{i+1}\cdots v_{i+j} = u \}
\]
Thus, the definition of $S_{i,u}$ depends on $j$ but we omit the suffix $j$. 
The swapping lemma asserts that if the ratio $|S_{i,u}|/|S|$ is small enough 
for any $i,j,u$ (with certain properties) 
then there exist two strings $x, y \in S$ such that $x^\prime$ and $y^\prime$ belong to $L$, 
where strings $x^\prime$ and $y^\prime$ are obtained by swapping the midsections of $x$ and $y$, 
and such that the midsections are different. 

\begin{lemma}[The Swapping Lemma for Context-free languages \cite{Y2009}]  
Suppose that $\Sigma $ has at least two letters, and that $L$ is an infinite context-free language over $\Sigma$. 
Then, there exists a positive integer $m$, \emph{a swapping lemma constant}, with the following properties. 

Suppose that $n \geq 2$ is a natural number, $S$ is a subset of $L \cap \Sigma^n$,  
and $j_0, k$ are natural numbers such that $2 \leq j_0$, $2j_0 \leq k \leq n$, 
and such that for any positive integer $i \leq n - j_0$ and any string $u \in \Sigma^{j_0}$, 
we have: 
\[
|S_{i,u}| < |S|/m(k - j_0 +1)(n - j_0 +1)
\]
Then, there exist positive integers 
$i, j$ and two strings $x=x_1 x_2 x_3, y = y_1 y_2 y_3 \in S$ with the following properties: 
$i + j \leq n$, $j_0 \leq j \leq k$, $|x_1|=|y_1|=i$, $|x_2|=|y_2|=j$, $|x_3|=|y_3|$, 
$x_2 \ne y_2$, $x_1 y_2 x_3 \in L$, and $y_1 x_2 y_3 \in L$. 
\end{lemma}

\section{Our Test Language} \label{section:testlanguage}

\begin{definition}
Suppose that we regard each natural number as a letter. 
Suppose $\Sigma $ is a finite subset of $\mathbb{N}$ and $w=w_1 \cdots w_n$ 
is a string over $\Sigma$. 
Given a natural number $c$, let $\Sigma_{\times c}$ denote $\{ c \times m: m \in \Sigma \}$. 
We let $(w)_{\times c}$ denote the string over $\Sigma_{\times c}$ 
given by each component of $w$ multiplied by $c$. 
\[
(w_1 \cdots w_n)_{\times c} = u_1 \cdots u_n, 
\mbox{ where } u_j = c \times w_j \mbox{, for each $j$}
\]
\qed
\end{definition}

For example, $(1211)_{\times 3} = 3633$. 

\begin{definition} \label{definition:l2}
We define our test language $L_2$ as follows (the suffix 2 is that of CFL(2)). 
\[
L_2 := \{ w (w^R)_{\times 3} (w)_{\times 15} (w^R)_{\times 5} : w \in \{ 1,2 \}^+ \}
\]
\qed
\end{definition}

\section{Main Theorem and Its Proof} \label{section:main}

\begin{lemma} \label{lemma:1}
$L_2$ belongs to $\mathrm{CFL}(2)$. 
\end{lemma}

\begin{proof}
We define languages $L_{2,1}$ and $L_{2,2}$ as follows. These are clearly context-free languages.
\begin{align*}
L_{2,1} &:= \{ w (w^R)_{\times 3} x : w \in \{ 1,2 \}^+, x \in \{ 5, 10, 15, 30 \}^+ \}
\\
L_{2,2} &:= \{ y (y^R)_{\times 5} : y \in \{ 1,2,3,6 \}^+ \}
\end{align*}

Then it holds that $L_2 = L_{2,1} \cap L_{2,2}$.
Hence $L_2$ belongs to CFL(2). 
\end{proof}

\begin{lemma} \label{lemma:2}
$L_2$ is $\mathrm{CFL}$-immune. 
\end{lemma}

\begin{proof}
$L_2$ is a subset of $L_2^\prime := \{ wxy : |w|=|x|, 2|w|=|y|, w \in \{ 1,2 \}^+, x \in \{ 3,6 \}^+  \mbox{ and } y \in \{ 5, 10, 15, 30 \}^+
\}$. 
Let $L_2^{\prime\prime}  := \{ a^n b^n c^{2n} : n \in \mathbb{N} \}$. 
A standard argument based on Bar-Hillel's lemma shows that $L_2^{\prime\prime}$ is CFL-immune, and a similar argument shows that $L_2^\prime $ is CFL-immune. 
Hence, $L_2$ is CFL-immune. 
\end{proof}

\begin{theorem} \label{theorem:main} (Main theorem)
There exists a $\mathrm{CFL}$-immune set in $\mathrm{CFL}(2) - \mathrm{CFL}/n$.
\end{theorem}

\begin{proof}
By Lemmas~\ref{lemma:1} and \ref{lemma:2}, it is sufficient to show that $L_2$ does not belong to CFL/$n$. 
We work with Yamakami's swapping lemma for context-free languages \cite{Y2009}. 
Consult Example 4.2 of \cite{Y2009} for a basic usage of the swapping lemma. 
For a proof by contradiction, fix a function $h$ and a context-free language $L$ such that 
$\forall n ~ |h(n)| = n$,  
and such that for any $\xi \in \{ 1,2,3,6,5,10,15,30 \}^*$, the following holds.
\[
\xi \in L_2 \,\, \leftrightarrow \,\, \begin{bmatrix} \xi \\ h(| \xi |) \end{bmatrix} \in L
\]

Let $m$ be a swapping lemma constant for the context-free language $L$. 
Let $n$ be a multiple of 16 with the following property. 
\begin{equation} \label{eq:main:1}
2^{n/4} > (2mn^2)^4
\end{equation}

We define a subset $S$ of $L$ as follows. 
\[
S := \left\{ \begin{bmatrix} \xi \\ h(n) \end{bmatrix} \in L : \xi \in \{ 1,2,3,6,5,10,15,30 \}^n \right\}
\]

Since the string $w (w^R)_{\times 3} (w)_{\times 15} (w^R)_{\times 5}$ is uniquely determined by 
$w \in \{ 1,2 \}^+$, the following holds.
\begin{equation} \label{eq:main:2}
|S| = 2^{n/4}
\end{equation}

Let $k$ and $j_0$ be the followings. Here, the base of the logarithm is 2.
\begin{align}
k &:= n/4  \label{eq:main:4} 
\\
j_0 &:= 2(\lceil \log (mn^2) \rceil +1)  \label{eq:main:5} 
\end{align}

By \eqref{eq:main:1} and \eqref{eq:main:4}, 
$k = n/4 > 4 ( \log (mn^2) + 1) $. 
Since $n/4$ is a multiple of 4, $k \geq 4 ( \lceil \log (mn^2) \rceil + 1) $. 
By \eqref{eq:main:5}, we get the following. 

\begin{equation} \label{eq:main:6}
k \geq 2 j_0
\end{equation}

Given a natural number $i$ and a string $u$ over the alphabet of $L$ 
such that $i + j_0 \leq n$ and $|u|=j_0$, we define $S_{i,u}$ as follows. 
\[
S_{i,u} = \{ v_1 \cdots v_n \in S : v_{i+1}\cdots v_{i+j_0} = u \}
\]

When $w (w^R)_{\times 3} (w)_{\times 15} (w^R)_{\times 5}$ is in $S_{i,u}$, 
some bits of $w$ are bound by $u$. 
We are going to estimate the number of bits bound by $u$. 
The minimal number is achieved when $u$ 
spans the border of any two blocks with the center of $u$ at the border. 
Thus, at least $\lceil j_0 / 2 \rceil$ bits of $w$ are bound by $u$. 
Therefore, we have the following. 
\begin{equation} \label{eq:main:3}
|S_{i,u}| \leq 2^{n/4 - j_0/2} 
\end{equation}

Now, we have the following.
\begin{equation} \label{eq:main:7}
|S_{i,u}| < |S| / kmn
\end{equation}

This is shown as follows. 
$2^{j_0 / 2} \geq 2mn^2$  [by \eqref{eq:main:5}] ~ 
$= 8kmn$ [by \eqref{eq:main:4}] ~ 
$> kmn$. Thus, $2^{j_0 / 2} > kmn$. 
By  \eqref{eq:main:2} and  \eqref{eq:main:3}, we have shown \eqref{eq:main:7}. 

By \eqref{eq:main:6} and \eqref{eq:main:7}, we can apply the swapping lemma for context-free languages \cite[Lemma 4.1]{Y2009} to the present setting.

By the swapping lemma, there exist natural numbers $i,j$ and strings $x,y \in S$ with the following properties. 
\begin{itemize}
\item 
$1 \leq i \leq n -j$ and $j_0 \leq j \leq k (= n/4)$

\item
$x,y$ are of the form $x = x_1 x_2 x_3, y= y_1 y_2 y_3$, where each $x_\ell$ and $y_\ell$ are strings, 
and it holds that $|x_1|=|y_1|=i$, $|x_2|=|y_2|=j$, $|x_3|=|y_3|$, $x_2 \ne y_2$, 
$x_1 y_2 x_3 \in L$ and $y_1 x_2 y_3 \in L$.
\end{itemize}

Let $\xi, \eta, \xi_\ell$ and $\eta_\ell$ ($\ell = 1,2,3$) be the projections of 
$x, y, x_\ell$ and $y_\ell$ to the first track, respectively. 
For example, 
\(
x = \begin{bmatrix} \xi \\ h(n) \end{bmatrix}, 
y = \begin{bmatrix} \eta \\ h(n) \end{bmatrix}
\).

Since $x$ belongs to $L$ and the second component is $h(n)$, it holds that $\xi \in L_2$. Therefore $\xi$ is of the form $w (w^R)_{\times 3} (w)_{\times 15} (w^R)_{\times 5}$ for some string $w$ of length $n/4$. Here, it holds that $|\xi_2|=j \leq k = n/4 = |w|$. 
Therefore, $\xi_2$ is included by either one of $w$, $(w^R)_{\times 3}$, $(w)_{\times 15}$ and $(w^R)_{\times 5}$ or included by consecutive two of them. 
The same holds for $\eta_2$. 

Fig.~\ref{fig:iplusjleqn4} demonstrates the case of $i+j \leq n/4$. 
Here, $\xi_2$ is included by $w$. 
Fig.~\ref{fig:ilessn4lessiplusj} demonstrates the case of $i < n/4 < i+j$. 
Here, $\xi_2$ is included by $w (w^R)_{\times 3}$. The other cases are similar. 

The swapping of $x_2$ and $y_2$ does not affect the second components $h(n)$ of $x$ and $y$ 
(Fig.~\ref{fig:swapnotaffect}). 
Thus both $\xi_1 \xi_2 \xi_3$ and $\xi_1 \eta_2 \xi_3$ belong to $L_2$, and $\xi_2 \ne \eta_2$. 
Hence, we get a contradiction. 

\begin{figure}[htb]
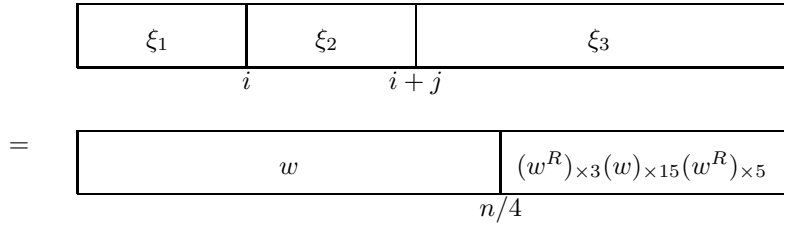

\begin{center}
\begin{tabular}{p{2em}p{2em}p{2em}p{2em}p{2em}p{2em}p{2em}p{5em}}
\cline{2-8}
 & 
\multicolumn{2}{|c|}{} & 
\multicolumn{2}{c|}{} & 
\multicolumn{3}{c}{} 
\\
 & 
\multicolumn{2}{|c|}{\raisebox{4pt}[0cm][0cm]{\quad $\xi_1$ \quad}} & 
\multicolumn{2}{c|}{\raisebox{4pt}[0cm][0cm]{\quad $\xi_2$ \quad}} & 
\multicolumn{3}{c}{\raisebox{4pt}[0cm][0cm]{$\xi_3$}} 
\\
\cline{2-8}
& & \multicolumn{2}{c}{$i$} & \multicolumn{2}{c}{$i+j$} &  &  
\\
 & & & & & & & 
\\
\cline{2-8}
= &  \multicolumn{5}{|c|}{} & 
\multicolumn{2}{c}{} 
\\
 &  \multicolumn{5}{|c|}{\raisebox{4pt}[0cm][0cm]{$w$}} & 
\multicolumn{2}{c}{\raisebox{4pt}[0cm][0cm]{$(w^R)_{\times 3} (w)_{\times 15} (w^R)_{\times 5}$}} 
\\
\cline{2-8}
 & & & & &\multicolumn{2}{c}{$n/4$} &  
\end{tabular}
\caption{The case of $i+j \leq n/4$ \label{fig:iplusjleqn4}}
\end{center}
\end{figure}

\begin{figure}[htb]
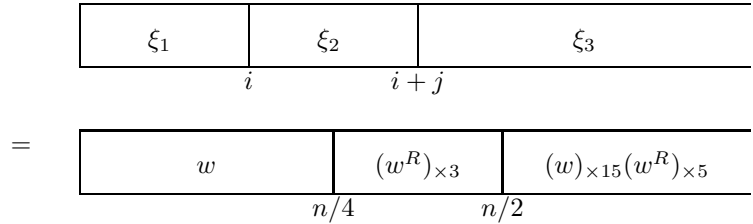

\begin{center}
\begin{tabular}{p{2em}p{2em}p{2em}p{2em}p{2em}p{2em}p{2em}p{5em}}
\cline{2-8}
 & 
\multicolumn{2}{|c|}{} & 
\multicolumn{2}{c|}{} & 
\multicolumn{3}{c}{} 
\\
 & 
\multicolumn{2}{|c|}{\raisebox{4pt}[0cm][0cm]{\quad $\xi_1$ \quad}} & 
\multicolumn{2}{c|}{\raisebox{4pt}[0cm][0cm]{\quad $\xi_2$ \quad}} & 
\multicolumn{3}{c}{\raisebox{4pt}[0cm][0cm]{$\xi_3$}} 
\\
\cline{2-8}
& & \multicolumn{2}{c}{$i$} & \multicolumn{2}{c}{$i+j$} &  &  
\\
 & & & & & & & 
\\
\cline{2-8}
= &  \multicolumn{3}{|c|}{} & 
\multicolumn{2}{c|}{} & 
\multicolumn{2}{c}{} 
\\
 &  \multicolumn{3}{|c|}{\raisebox{4pt}[0cm][0cm]{$w$}} & 
 \multicolumn{2}{|c|}{\raisebox{4pt}[0cm][0cm]{$(w^R)_{\times 3}$}} & 
\multicolumn{2}{c}{\raisebox{4pt}[0cm][0cm]{$(w)_{\times 15} (w^R)_{\times 5}$}} 
\\
\cline{2-8}
 & & & \multicolumn{2}{c}{$n/4$} &\multicolumn{2}{c}{$n/2$} &  
\end{tabular}
\caption{The case of $i < n/4 < i+j$ \label{fig:ilessn4lessiplusj}}
\end{center}
\end{figure}

\begin{figure}[H]
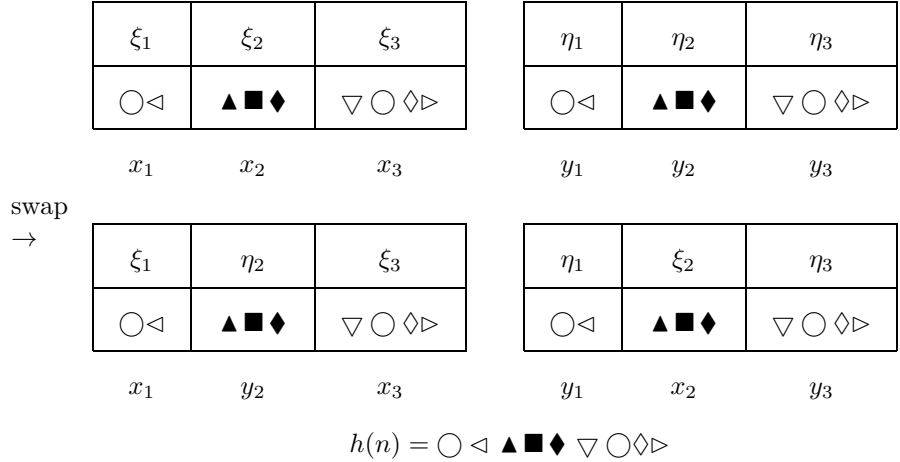

\begin{center}
\begin{tabular}{p{2.5em}p{2.5em}p{3.5em}p{4.5em}p{1em}p{2.5em}p{3.5em}p{4.5em}}
\cline{2-4} \cline{6-8}
&
\multicolumn{1}{|c|}{\raisebox{4pt}[0cm][0cm]{}} & 
\multicolumn{1}{c|}{\raisebox{4pt}[0cm][0cm]{}} & 
\multicolumn{1}{c|}{\raisebox{4pt}[0cm][0cm]{}} 
& & 
\multicolumn{1}{|c|}{\raisebox{4pt}[0cm][0cm]{}} & 
\multicolumn{1}{c|}{\raisebox{4pt}[0cm][0cm]{}} & 
\multicolumn{1}{c|}{\raisebox{4pt}[0cm][0cm]{}} \\
&
\multicolumn{1}{|c|}{\raisebox{4pt}[0cm][0cm]{$\xi_1$}} & 
\multicolumn{1}{c|}{\raisebox{4pt}[0cm][0cm]{$\xi_2$}} & 
\multicolumn{1}{c|}{\raisebox{4pt}[0cm][0cm]{$\xi_3$}} 
& & 
\multicolumn{1}{|c|}{\raisebox{4pt}[0cm][0cm]{$\eta_1$}} & 
\multicolumn{1}{c|}{\raisebox{4pt}[0cm][0cm]{$\eta_2$}} & 
\multicolumn{1}{c|}{\raisebox{4pt}[0cm][0cm]{$\eta_3$}} \\
\cline{2-4} \cline{6-8} 
&
\multicolumn{1}{|c|}{\raisebox{4pt}[0cm][0cm]{}} & 
\multicolumn{1}{c|}{\raisebox{4pt}[0cm][0cm]{}} & 
\multicolumn{1}{c|}{\raisebox{4pt}[0cm][0cm]{}} 
& & 
\multicolumn{1}{|c|}{\raisebox{4pt}[0cm][0cm]{}} & 
\multicolumn{1}{c|}{\raisebox{4pt}[0cm][0cm]{}} & 
\multicolumn{1}{c|}{\raisebox{4pt}[0cm][0cm]{}} \\
&
\multicolumn{1}{|c|}{\raisebox{4pt}[0cm][0cm]{$\bigcirc \lhd$}} & 
\multicolumn{1}{c|}{\raisebox{4pt}[0cm][0cm]{$\blacktriangle \, \blacksquare \, \blacklozenge$}} & 
\multicolumn{1}{c|}{\raisebox{4pt}[0cm][0cm]{$\bigtriangledown \bigcirc \lozenge \rhd $}} 
& & 
\multicolumn{1}{|c|}{\raisebox{4pt}[0cm][0cm]{$\bigcirc \lhd$}} & 
\multicolumn{1}{c|}{\raisebox{4pt}[0cm][0cm]{$\blacktriangle \, \blacksquare \, \blacklozenge$}} &
\multicolumn{1}{c|}{\raisebox{4pt}[0cm][0cm]{$\bigtriangledown \bigcirc \lozenge \rhd $}} 
\\
\cline{2-4} \cline{6-8} 
 & & & & & & & 
\\
&
\multicolumn{1}{c}{\raisebox{4pt}[0cm][0cm]{$x_1$}} & 
\multicolumn{1}{c}{\raisebox{4pt}[0cm][0cm]{$x_2$}} & 
\multicolumn{1}{c}{\raisebox{4pt}[0cm][0cm]{$x_3$}} & 
&
\multicolumn{1}{c}{\raisebox{4pt}[0cm][0cm]{$y_1$}} & 
\multicolumn{1}{c}{\raisebox{4pt}[0cm][0cm]{$y_2$}} & 
\multicolumn{1}{c}{\raisebox{4pt}[0cm][0cm]{$y_3$}} 
\\
 swap & & & & & & & 
\\
\cline{2-4} \cline{6-8}
$\rightarrow$
&
\multicolumn{1}{|c|}{\raisebox{4pt}[0cm][0cm]{}} & 
\multicolumn{1}{c|}{\raisebox{4pt}[0cm][0cm]{}} & 
\multicolumn{1}{c|}{\raisebox{4pt}[0cm][0cm]{}} 
& & 
\multicolumn{1}{|c|}{\raisebox{4pt}[0cm][0cm]{}} & 
\multicolumn{1}{c|}{\raisebox{4pt}[0cm][0cm]{}} & 
\multicolumn{1}{c|}{\raisebox{4pt}[0cm][0cm]{}} \\
&
\multicolumn{1}{|c|}{\raisebox{4pt}[0cm][0cm]{$\xi_1$}} & 
\multicolumn{1}{c|}{\raisebox{4pt}[0cm][0cm]{$\eta_2$}} & 
\multicolumn{1}{c|}{\raisebox{4pt}[0cm][0cm]{$\xi_3$}} 
& & 
\multicolumn{1}{|c|}{\raisebox{4pt}[0cm][0cm]{$\eta_1$}} & 
\multicolumn{1}{c|}{\raisebox{4pt}[0cm][0cm]{$\xi_2$}} & 
\multicolumn{1}{c|}{\raisebox{4pt}[0cm][0cm]{$\eta_3$}} \\
\cline{2-4} \cline{6-8} 
&
\multicolumn{1}{|c|}{\raisebox{4pt}[0cm][0cm]{}} & 
\multicolumn{1}{c|}{\raisebox{4pt}[0cm][0cm]{}} & 
\multicolumn{1}{c|}{\raisebox{4pt}[0cm][0cm]{}} 
& & 
\multicolumn{1}{|c|}{\raisebox{4pt}[0cm][0cm]{}} & 
\multicolumn{1}{c|}{\raisebox{4pt}[0cm][0cm]{}} & 
\multicolumn{1}{c|}{\raisebox{4pt}[0cm][0cm]{}} \\
&
\multicolumn{1}{|c|}{\raisebox{4pt}[0cm][0cm]{$\bigcirc \lhd$}} & 
\multicolumn{1}{c|}{\raisebox{4pt}[0cm][0cm]{$\blacktriangle \, \blacksquare \, \blacklozenge$}} & 
\multicolumn{1}{c|}{\raisebox{4pt}[0cm][0cm]{$\bigtriangledown \bigcirc \lozenge \rhd $}} 
& & 
\multicolumn{1}{|c|}{\raisebox{4pt}[0cm][0cm]{$\bigcirc \lhd$}} & 
\multicolumn{1}{c|}{\raisebox{4pt}[0cm][0cm]{$\blacktriangle \, \blacksquare \, \blacklozenge$}} &
\multicolumn{1}{c|}{\raisebox{4pt}[0cm][0cm]{$\bigtriangledown \bigcirc \lozenge \rhd $}} 
\\
\cline{2-4} \cline{6-8} 
 & & & & & & & 
\\
&
\multicolumn{1}{c}{\raisebox{4pt}[0cm][0cm]{$x_1$}} & 
\multicolumn{1}{c}{\raisebox{4pt}[0cm][0cm]{$y_2$}} & 
\multicolumn{1}{c}{\raisebox{4pt}[0cm][0cm]{$x_3$}} & 
&
\multicolumn{1}{c}{\raisebox{4pt}[0cm][0cm]{$y_1$}} & 
\multicolumn{1}{c}{\raisebox{4pt}[0cm][0cm]{$x_2$}} & 
\multicolumn{1}{c}{\raisebox{4pt}[0cm][0cm]{$y_3$}} 
\end{tabular}

\vspace{.5\baselineskip}

\hspace{4.5em} $h(n) = \bigcirc \lhd \, \blacktriangle \, \blacksquare \, \blacklozenge \, \bigtriangledown \bigcirc \lozenge \rhd $

\caption{The swapping does not affect the second track \label{fig:swapnotaffect}}
\end{center}
\end{figure}

Thus, we have shown that $L_2$ does not belong to CFL/$n$. Hence, we have shown the theorem. 
\end{proof}

\section{Serial Advices} \label{section:serial}

\subsection{A Variation of the Main Theorem}

Damm and Holzer \cite{DH1995} investigated advised language classes earlier than Tadaki et al. \cite{TYL2005}. In the definition of Damm and Holzer, the arrangement of the advice and the input is serial, while in that of Tadaki et al., it is parallel.
In this section, we show that the same result as our main theorem holds for the advised language class in the sense of Damm and Holzer. 

\begin{definition}  \label{def:cslashndh} (Damm and Holzer \cite{DH1995})
Given a class ${\cal C}$ of languages, the advised language class ${\cal C}/n$ in the sense of Damm and Holzer is defined as follows. 
Suppose that $\Sigma$ and $\Gamma_0$ are alphabets. 
Suppose that $L$ is a language over $\Sigma$. 
A language $L$ belongs to ${\cal C}/n$ (in the sense of Damm and Holzer) if and only if there exist a language $L^{\prime\prime} \in {\cal C}$ and a function $g : \mathbb{N} \to \Gamma_0^\ast$ such that 
$\forall n ~ |g(n)|=n$, and such that the following holds. 
\[
x \in L \,\, \Leftrightarrow \,\, g(|x|) ~x \in L^{\prime\prime}
\]

Here, $g(|x|) ~x$ is the concatenation of $g(|x|)$ and $x$.    
Then, $g$ is called an \emph{advice function}. $g(n)$ is the advice at length $n$. 
\qed 
\end{definition}

In the remainder of the paper, ${\cal C}/n$ denotes the advised class in the sense of Tadaki, Yamakami and Lin \cite{TYL2005}, that is, the parallel one defined in Introduction. 
On the other hand, $({\cal C}/n)_\mathrm{serial}$ denotes the advised class defined in this section. 
The following is a variation of our main theorem.

\begin{theorem} \label{theorem:maindh}
There exists a $\mathrm{CFL}$-immune set in $\mathrm{CFL}(2) - (\mathrm{CFL}/n)_\mathrm{serial}$.
\end{theorem}

\begin{proof}
Let $L_2$ be the language defined in Definition~\ref{definition:l2}. 
In the same way as the proof of Theorem~\ref{theorem:main}, 
it is sufficient to show that $L_2$ does not belong to $(\mathrm{CFL}/n)_\mathrm{serial}$. 
Fix an advice function $g$ and a context-free language $L$ such that for any 
$x \in \{ 1,2,3,6,5,10,15,30 \}^*$, the following holds. 
\[
x \in L_2 \,\, \leftrightarrow \,\, g(|x|)~x \in L
\]

Choose natural numbers $m, n, k, j_0$ in the exactly same way as the proof of Theorem~\ref{theorem:main}. 
We define a subset $S$ of $L$ as follows. 
\[
S := \left\{ g(n)~x \in L : x \in \{ 1,2,3,6,5,10,15,30 \}^n \right\}
\]

Thus, each member of $S$ has length $2n$. 
We have $|S_{i,u}| < |S| / km(2n)$. 
By the swapping lemma, there exist natural numbers $i,j$ and strings $x$ and $y$ of the following properties. 
$g(n)~x$ and $g(n)~y$ are in $S$ and they are of the form $g(n)~x = x_1 x_2 x_3$ and $g(n)~y= y_1 y_2 y_3$, where each $x_\ell$ and $y_\ell$ are strings, 
and it holds that $|x_1|=|y_1|$, $|x_2|=|y_2| \leq n/4$, $|x_3|=|y_3|$, $x_2 \ne y_2$, 
$x_1 y_2 x_3 \in L$ and $y_1 x_2 y_3 \in L$.

Since $x_2$ and $y_2$ are not identical, they are not substrings of $g(n)$. 
Hence, in the same way as the proof of Theorem~\ref{theorem:main}, we get a contradiction. 

Thus, we have shown that $L_2$ does not belong to $(\mathrm{CFL}/n)_\mathrm{serial}$. 
Hence, we have shown the theorem.
\end{proof}

\subsection{Remarks and Problems}

In this subsection, we discuss separations of parallel advice classes and serial advice classes. 

In general, the two concepts of advised classes do not coincide. 
Recall that REG denotes the class of all regular languages. 

\begin{example} \label{example:regslashndifferent1}
$\mathrm{REG}/n$ is not a subset of $(\mathrm{REG}/n)_\mathrm{serial}$. 
A proof is as follows. 
Damm and Holzer show that the language $L_\mathrm{eq} = \{ 0^n 1^n : n \in \mathbb{N} \}$ does not belong to $(\mathrm{REG}/n)_\mathrm{serial}$ \cite[Propositions 1 and 7]{DH1995}. 
On the other hand, let $h(n)$ be $0^{n/2}1^{n/2}$ if $n$ is even; $2^n$ otherwise. 
Then, by means of an advice function $h$, $L_\mathrm{eq}$ is shown to be in $\mathrm{REG}/n$. 
\qed
\end{example}

\begin{example} \label{regslashnsubset}
$(\mathrm{REG}/n)_\mathrm{serial}$ is a subset of $\mathrm{REG}/n$. 
A proof is as follows. 
Suppose $L$ is an element of $(\mathrm{REG}/n)_\mathrm{serial}$. Let $L^{\prime\prime}$ and $g$ be a regular language and an advice function satisfying the requirements in Definition~\ref{def:cslashndh}. Let $M$ be a deterministic finite automaton that accepts $L^{\prime\prime}$. 
Given a natural number $n$, let $q_{(n)}$ be the state when $M$ has read $g(n)$. 
Provided that $q_{(n)}$ is given, without knowing what $g(n)$ is, we can simulate the moves of $M$ after reading $g(n)$.  
Then, we define $h(n)$ as to be $q_{(n)}~0^{n-1}$.  
The 0s in the tail are just for adjusting the length of $h(n)$. 
Let $\Gamma$ be the union of $\{ 0 \}$ and the set of the states of $M$. 
Now, it is easy to define a regular language $L^\prime$ satisfying the requirements in Definition~\ref{def:cslashntyl} 
with respect to $\Gamma $ and $h$. 
Therefore, $L$ belongs to $\mathrm{REG}/n$. 
\qed
\end{example}

A direct proof of Example~\ref{example:regslashndifferent1} is given by means of prefix-free Kolmogorov complexity. 

\begin{definition} \cite{N2009}
\begin{itemize}
\item 
A string $u=u_1\cdots u_m$ is a \emph{prefix} of a string $v=v_1 \cdots v_n$ if $m \leq n$ and for each $i \leq m$ it holds that $u_i = v_i$. 
A set $S$ of strings is \emph{prefix-free} if for each $u, v \in S$ such that $u \ne v$, $u$ is not a prefix of $v$. A partial function $M: \{ 0,1 \}^\ast \to \{ 0, 1 \}$ is called a \emph{prefix-free machine} if $M$ is a partial recursive function and the domain of $M$ is a prefix-free set. 

\item 
For a prefix-free machine $M$, its \emph{descriptive complexity} $K_M: \{ 0,1 \}^\ast \to \mathbb{N} \cup \{ \infty \}$ is defined as follows. Suppose $x \in \{ 0,1 \}^\ast$. If there exists $\sigma \in \{ 0,1 \}^\ast$ such that $M(\sigma )=x$ then $K_M (x)$ is the length of a shortest such $\sigma$. 
If there is no such $\sigma$ then $K_M (x)$ is $\infty$.

\item
A prefix-free machine $R$ is an \emph{optimal prefix-free machine}\/ if for each prefix-free machine $M$,  there is a constant $d$ (depending on $M$) such that for each $x \in \{ 0,1 \}^\ast$, 
$K_R(x) \leq K_M(x) + d$.
\end{itemize}
\qed
\end{definition}

It is known that there exists an optimal prefix-free machine \cite[Proposition 2.2.7]{N2009}. 
We fix such a machine $R$, and let $K$ denote $K_R$. 
For an infinite binary string $Z : \mathbb{N} \to \{ 0,1 \}$ and a natural number $n$, we let $Z \upharpoonright n$ denote $Z(0)Z(1)\cdots Z(n-1)$. 
It is known that there exists a binary string $Z$ of the following property \cite[section 3.2]{N2009}.
\begin{equation} \label{equation:chaitinrandom}
\exists b \in \mathbb{N} ~ \forall n \in \mathbb{N} \,\, [K(Z \upharpoonright n) > n -b]
\end{equation}

Example~\ref{example:regslashndifferent2} is a refinement of the proof by Damm and Holzer \cite{DH1995} that $L_\mathrm{eq} = \{ 0^n 1^n : n \in \mathbb{N} \} \not\in (\mathrm{REG}/n)_\mathrm{serial}$. 

\begin{example} \label{example:regslashndifferent2}
A direct proof that $\mathrm{REG}/n$ is not a subset of $(\mathrm{REG}/n)_\mathrm{serial}$. 
Let $Z$ be an infinite binary string satisfying \eqref{equation:chaitinrandom}. 
Let $L := \{ Z \upharpoonright n : n \in \mathbb{N} \}$. 

By means of an advice function $h(n) = Z \upharpoonright n$, $L$ is shown to be in $\mathrm{REG}/n$.

We are going to show that $L$ does not belong to $(\mathrm{REG}/n)_\mathrm{serial}$. 
Assume that $L$ belongs to it. Suppose that $L^{\prime\prime}$ and $g$ are a regular language and an advice function satisfying the requirements in Definition~\ref{def:cslashndh}. 
Suppose that $M$ is a deterministic finite automaton that accepts $L^{\prime\prime}$. 
Let $Q$ be its set of states. 
For each $n$, let $q_{(n)}$ be the state when $M$ has read $g(n)$.

We define a deterministic Turing machine $N$ as follows. 
An input is an ordered pair $(q,n) \in Q \times \mathbb{N}$. 
For each $y \in \{ 0,1 \}^n$, simulate the moves of $M$ as follows. 
Set the state (of the virtual $M$) being $q$. Let $M$ read $y$. If $M$ accepts $y$, return $y$ and halt. 

If the for-loop finishes without any output, then $N$ does not halt. 

By the definition of $L^{\prime\prime}$ and $g$, for the input $(q_{(n)}, n)$, $N$ outputs $Z \upharpoonright n$. In addition, by a certain appropriate coding, we may assume that the domain of $N$ is a prefix free set. 
For example, code an ordered pair $(u_1 \cdots u_m, v_1 \cdots v_\ell)$ by a string $u_1 u_1 \cdots u_m u_m 01 v_1 v_1 \cdots v_\ell v_\ell 01$. 

Therefore, $ K_N (Z \upharpoonright n)$ is in the order of the length of $(q_{(n)},n)$. 
Thus, it is $O ( \log_2 (n))$. 
Hence, by the definition of an optimal machine, 
$K(Z \upharpoonright n) \leq K_N (Z \upharpoonright n) + O(1) = O(\log_2 n)$. 
This contradicts to the assumption of \eqref{equation:chaitinrandom}. 
\qed
\end{example}

To our knowledge, we do not know whether the following hold. 

\begin{enumerate}
\item $\mathrm{CFL}/n \subset (\mathrm{CFL}/n)_\mathrm{serial}$~?
\item $(\mathrm{CFL}/n)_\mathrm{serial} \subset \mathrm{CFL}/n$~?
\end{enumerate}
 
\section*{Acknowledgment}

The author would like to thank Tomoyuki Yamakami, Masahiro Kumabe and Yuki Mizusawa for helpful discussions. 


\end{document}